\documentclass[conference,10pt]{IEEEtran}
\IEEEoverridecommandlockouts
\usepackage{cite}
\usepackage{amsmath,amssymb,amsfonts}
\usepackage{algorithmic}
\usepackage{graphicx}
\usepackage{textcomp}
\usepackage{xcolor}
\def\BibTeX{{\rm B\kern-.05em{\sc i\kern-.025em b}\kern-.08em
		T\kern-.1667em\lower.7ex\hbox{E}\kern-.125emX}}

\usepackage{amsmath,amsthm}
\usepackage{amssymb}
\usepackage{graphicx,setspace}
\usepackage{amsmath,amssymb}
\usepackage{color}
\usepackage[english]{babel}
\usepackage{mathrsfs,dsfont}
\usepackage{array}
\usepackage{subcaption}
\usepackage{tikz}
\usepackage{adjustbox}
\usepackage{multirow}
\usetikzlibrary{arrows, automata, decorations.pathmorphing, backgrounds, positioning, fit, petri, matrix, patterns}
\usetikzlibrary{shapes,snakes}

\usepackage[linesnumbered,vlined,ruled]{algorithm2e}

\usepackage{mathrsfs}
\usepackage[mathscr]{euscript}

\newtheorem{theorem}{\bf Theorem}
\newtheorem{lemma}[theorem]{\bf Lemma}
\newtheorem{corollary}[theorem]{\bf Corollary}

\newtheorem{assumption}{Assumption}

\newtheorem{remark}{Remark}
\newtheorem{definition}{Definition}

\def\QED{~\rule[-1pt]{5pt}{5pt}\par\medskip}
\renewenvironment{proof}{{\bf Proof: \ }}{ \hfill \QED}

\renewcommand{\Re}{\mathbb{R}}

\usepackage{graphicx}
\usepackage{caption}
\usepackage{subcaption}

\newcommand{\sumt}{\sum_{t=1}^{T}}

\makeatletter
\def\endthebibliography{%
	\def\@noitemerr{\@latex@warning{Empty `thebibliography' environment}}%
	\endlist
}
\makeatother

\begin{document}
	
	\title{Interference Constrained  Beam Alignment for Time-Varying Channels via Kernelized Bandits}

	\author{
		\IEEEauthorblockN{Yuntian Deng}
		\IEEEauthorblockA{Department of ECE \\
			The Ohio State University\\
			Columbus, OH 43210 \\
			deng.556@osu.edu}
		\and
		\IEEEauthorblockN{Xingyu Zhou}
		\IEEEauthorblockA{Department of ECE  \\
			Wayne State University\\
			Detroit, MI 48202\\
			xingyu.zhou@wayne.edu}
		\and
		\IEEEauthorblockN{Arnob Ghosh, Abhishek Gupta}
		\IEEEauthorblockA{Department of ECE\\
			The Ohio State University\\
			Columbus, OH 43210 \\
			\{ghosh.244, gupta.706\}@osu.edu}
		\and
		\IEEEauthorblockN{Ness B. Shroff}
		\IEEEauthorblockA{Department of ECE \& CSE \\
			The Ohio State University\\
			Columbus, OH 43210 \\
			shroff.11@osu.edu}
	}
	
	\maketitle
	
	\begin{abstract}
		To fully utilize the abundant spectrum resources in millimeter wave (mmWave), Beam Alignment (BA) is necessary for large antenna arrays to achieve large array gains. In practical dynamic wireless environments, channel modeling is challenging due to time-varying and multipath effects. In this paper, we formulate the beam alignment problem as a non-stationary online learning problem with the objective to maximize the  received signal strength under interference constraint. In particular,  we employ the non-stationary kernelized bandit to leverage the correlation among beams and model the complex beamforming and multipath channel functions. Furthermore, to mitigate interference to other user equipment, we leverage the primal-dual method to design a constrained UCB-type kernelized bandit algorithm. Our theoretical analysis indicates that the proposed algorithm can adaptively adjust the beam in time-varying environments, such that both the cumulative regret of the received signal and constraint violations have sublinear bounds with respect to time. This result is of independent interest for applications such as adaptive pricing and news ranking. In addition, the algorithm assumes the channel is a black-box function and does not require any prior knowledge for dynamic channel modeling, and thus is applicable in a variety of scenarios. We further show that if the information about the channel variation is known, the algorithm will have better theoretical guarantees and performance. Finally, we conduct simulations to highlight the effectiveness of the proposed algorithm.
	\end{abstract}
	
	
	\section{Introduction}
	In the era of big data, the demand for high-speed communication is significantly increasing. Various data-hungry applications, such as virtual and augmented reality, high-resolution mobile video streaming, and delay-sensitive online gaming, come with significant traffic demands and incentivize existing cellular network to seek larger bandwidths by communicating at higher frequencies. Millimeter wave (mmWave) band, spanning from 30 to 300 GHz, is a promising technology to support such high demand. 
	%For example, one band in 5G NR in  24 GHz range or above can provide high channel bandwidth as large as 400 MHz.  
	
	There are three fundamental challenges which can hinder fully utilizing of the mmWave system: \textbf{(I)} \textbf{Sample complexity.} The propagation loss is severe due to high atmospheric attenuation, which increases at higher frequencies. To overcome this, beamforming with a large antenna array is leveraged to form a directional beam and reduce the channel loss, which could provide a higher throughput for a certain direction. Beam alignment between transmitter and receiver is therefore needed before data transmission  \cite{hashemi2018efficient}. Usually,  time-consuming exhaustive searches are performed and many samples are required to find the optimal beam direction.  \textbf{(II)} \textbf{Time-varying complicated function.} Signals reach the receiver via multiple paths, including one line-of-sight (LOS) path and several non-line-of-sight (NLOS) paths. To this end, the received signal strength is a non-linear non-convex function with respect to beamforming vectors, and the Channel State Information (CSI) for multi paths is difficult to obtain and estimate. In practical environments and mobile scenarios, the mmWave channel becomes time-varying and estimating the channel gain  becomes more challenging,  as existing paths may quickly disappear and new paths may appear \cite{zhang2020beam}. \textbf{(III)} \textbf{Interference Constraints.} Although the beam is directional, in heterogeneous environments with multiple transmitters, users in neighboring cells still have high interference if they use the same frequency \cite{bayraktar2021efficient}, particularly users who are located close to the cell's edge \cite{wang2019interference}. Therefore, some beam candidates cannot be used for data transmission. In dynamic environments, it is more difficult to check their eligibility with time-varying channels.

	%1) propagation loss
	%2) high interference
	%3) mobility CSI
	%4) multi path
	
	In the presence of these challenges, we pose the following question: \textit{In a  non-stationary multi-path environment,   can we learn to find the optimal interference-constrained beam alignment configuration in a sample-efficient manner? }
	%given the past samples of beam alignments,
	
	We cast this problem as a sequential decision-making problem, where the objective is to maximize the  received signal strength,  under a time-varying constraint on the interference at other users' terminals. We use the kernelized multi-armed bandits (i.e. Gaussian Process bandit) \cite{chowdhury2017kernelized} to model this problem, where each sample point represents a beamforming vector on an array of fixed-phase antennas. As different beamforming vector leads to various beam directions and interference levels, the received signal and the interference  change accordingly. We aim at designing a sequential query point selection (i.e., beamforming vector selection) strategy that maximizes the reward (i.e., received signal strength) in a time-varying channel with a small violation in the soft interference constraint.  
	
	Our constrained kernelized bandit model addresses the above three challenges in a unified way. Since the reward (received signal) function is continuous, it can be uniformly approximated by a function in the reproducing kernel Hilbert space (RKHS) \cite{micchelli2006universal} (under a proper choice of the kernel). To utilize this representation power of RKHS, we further assume that the reward function is within a RKHS induced by a kernel. With this assumption, the complicated multi-path channel in challenge \textbf{(II)} could be represented by one function within the RKHS. Such  kernel-based modeling is also used in WLAN \cite{kushki2007kernel} and wireless sensor networks \cite{honeine2009functional}. As the channel is time-varying, the reward function varies from time to time within the RKHS. Further, since similar beamforming vectors lead to similar beam and received signal strength, this assumption automatically utilizes the beam correlation information \cite{wu2019fast} as RKHS is kernel-based and correlation-embedded. Therefore, we significantly reduce the search space, compared to the exhaustive search in challenge \textbf{(I)}. 
	Moreover, we model the interference-based beam eligibility in challenge  \textbf{(III)} as a constraint function. When the received signal at each non-target user is less than a threshold, the constraint is satisfied and the beam is eligible. As the channel is time-varying, this constraint also changes with time. Finally, we develop the constrained kernelized bandit based on the primal-dual method and Gaussian Process agnostic setting. The proposed algorithm periodically restarts to overcome the non-stationarity in the environment. We further provide theoretical bounds on the performance of reward and constraint violations for any kernel.
	
	%Beamformer explained \cite{hamdy2020beamformers}  (replace it with a paper not company)
	
	In summary, our major contributions are as follows:
	
	\textbf{1) Algorithm Design:} Using the online learning framework, we propose a novel constrained kernelized bandit algorithm with restart for beam alignment in time-varying environments. In particular, we first leverage the inherent correlation and representation power of kernelized bandits, where the reward function supports general multi-path scenarios and various channel models. Second, we model the interference as a constrained bandit via the primal-dual method. Third, we propose a restart schedule to handle the time-varying environments. 
	
	\textbf{(2) Theoretical Results:} We provide theoretical guarantees on the cumulative regret and constraint violations, where the regret is the loss compared with the optimal algorithm. In particular, our proofs overcome technical challenges posed by time-varying environments. We further provide sublinear bounds for a general non-stationary constrained kernelized bandit framework, which is a missing part of the state-of-art kernelized bandit. Additionally, the proposed algorithm is performed without prior knowledge about total channel variation. If this knowledge is provided, we show that the algorithm can further improve performance and both regret and constraint violations bounds become tighter.

	\section{Related Work}
	\subsection{Beam Alignment}
	To reduce the beamforming overhead, there is a large amount of work on digital and analog beamforming \cite{roh2014millimeter} and channel selection \cite{combes2014dynamic}. The authors in \cite{hashemi2018efficient}  exploit the inherent correlation in beamforming and leverage this contextual information to reduce the search space for beam alignment. However, its assumption of unimodal function with respect to beam directions only holds in single path situations and fails in multi-path scenarios.  Wu et. al. \cite{wu2019fast} utilize the prior knowledge of the channel fluctuation to accommodate reward uncertainty. Specifically, they assume that the variance of the channel fluctuation is known and as a result,  less exploration is needed and the beam alignment is accelerated. The authors in \cite{zhang2020beam} leverage stationary stochastic bandits to sense the change of the environment in dynamic beamforming. In their model, actions (arms) are designed based on the difference (offset) of the indices of the optimal beams in two adjacent time slots, which measures the rate of change of the environments. They further assume that the reward of each arm is from a fixed distribution depending on the arm but not on the time slot (therefore, there exists a best arm that captures the best rate of changes), which limits its generalization to model dynamic environments. In \cite{gupta2020beam}, non-stationary multi-armed bandits are leveraged for time-varying channels in beam alignment, where they use a sliding window-based algorithm and weight penalty-based algorithm. However, in their model, arms are independent of each other and prior knowledge of the number of breakpoints (when the reward distribution changes) is required, which means it can only handle abruptly change scenarios rather than slowly-changing scenarios and the correlation between beams is not utilized. In contrast to previous works, we focus on the interference-constrained beam selection in a general time-varying environment, where changes are time-dependent and it includes both abruptly changes and slowly changes. The proposed algorithm can be performed without any prior knowledge about changes.
	
	%In \cite{hashemi2017out}, spatial correlations between sub-6 GHz and mmWave are leveraged to remedy the beamforming overhead of mmWave. Specifically, a coarse angle of arrival (AoA) estimation on sub-6 GHz is conducted and used in angular range scans for mmWave systems.
	\subsection{Non-stationary bandit}
	Multi-armed bandit (MAB) is a general framework in sequential decision-making, where the agent needs to select a query point at each time slot to maximize the cumulative reward \cite{robbins1952some}. The unknown reward function is usually a black-box function defined over a large domain space. In real-world scenarios, the reward function is often not fixed and varies over time. For example, in stochastic linear bandit (the black-box function is linear and query points are non-orthogonal), three common techniques are leveraged to tackle  time-varying rewards:  (i) periodically restarting the learning process \cite{zhao2020simple}, (ii) regression based on the data within a sliding window \cite{cheung2019learning}, (iii) putting a time-dependent weight across all data \cite{russac2019weighted}. Recently, kernelized  bandits (i.e. Gaussian process bandits) have become popular, as they generalize traditional MAB and linear bandit, by allowing the reward function to be non-linear and non-convex \cite{chowdhury2017kernelized}. Similarly, to tackle non-stationarity, there are mainly three methods:  restart \cite{zhou2021no}, sliding window \cite{zhou2021no} and  weight penalty \cite{deng2022weighted}. This problem becomes more challenging in practical applications where there are unknown soft constraints \cite{zhou2022kernelized}. In this paper, we generalize previous non-stationary results to the case with a general unknown time-varying constraint function. We further show that the restart method works with constraint bandits and cast the interference-constrained beam alignment problem as a constrained bandit model in time-varying environments.

	%For example, in stochastic linear bandit (the black-box function is linear and query points are non-orthogonal), four common techniques  are leveraged to tackle the time-varying:  periodically restart  the learning process \cite{zhao2020simple}, regression based on the data within a sliding window \cite{cheung2019learning}, putting a time-dependent weight across all data \cite{russac2019weighted}, adaptive restart based on stationary tests \cite{wei2021non}.

	%For example, in  Gaussian process bandits \cite{chowdhury2017kernelized} (the reward function is non-linear and non-convex, which generalizes traditional MAB and linear bandit), three common techniques  are leveraged to tackle the time-varying function: periodically restart  the learning process \cite{zhou2021no}, regression based on data points within a sliding window \cite{zhou2021no}, putting a time-dependent increasing weight across all data points \cite{deng2021weighted}. This problem becomes more challenging in practical applications where there are unknown soft constraints \cite{zhou2022kernelized}. In this paper, we generalize previous non-stationary results to the case with a general unknown time-varying constraint function. We further cast the interference-constrained beam alignment problem as a constrained bandit model, and derive theoretical bounds on cumulative rewards and constraint violations. 

	\section{Problem Statement}

	We consider a millimeter wave (mmWave) point-to-point communication system, where the transmitter (BS) is equipped with phased array antennas and the receiver (UE) uses a single antenna. Before data transmission, beam alignment is needed to achieve the array gains and high throughput. Due to the high power consumption in digital beamforming, we focus on an analog beamforming structure that relies on a single radio-frequency (RF) chain and one analog-to-digital converter (ADC) with less power consumption. In practical dynamic environments, the mmWave channel varies fast and the change becomes faster in mobile scenarios, where either the transmitter or the receiver is moving. 
	Let $x_t \in \mathbb C^{M}$ denote the beamforming weight vector at time $t$  with $M$ antennas. We assume that $x_t \in X$ is in the predefined codebook, where the number of beams $|X|$ is finite. By controlling $x_t$, we have various beams with different directions and widths. The transmitted pilot signal is $s$ for all time slots. 
	Since the receiver has the  omni-directional beam, the received signal at the receiver is the weighted combination of the message $s$ across all antennas at the transmitter.  For user equipment (UE) A, we have the noise-free received signal strength (RSS) as 
	\begin{align*}
		y_t = |H_t x_t s | =: f_t(x_t),
	\end{align*}
	where $H_t= \sum_{l=1}^L H_t^l$ denotes multipath time-varying channel gain for UE A with $L$ paths, which is a location-dependent unknown function of time $t$.  We assume that UE A observes a noisy signal strength with an additive Gaussian noise $n_t$, {\em without any estimates} of the channel gain $H_t$. Therefore, we have the observation at time $t$ as $r_t= f_t(x_t) +  n_t$.
	Besides UE A, we further consider a group of UEs $\mathcal J$. As UE A is the target the beam wants to focus on, UEs in $\mathcal J$ do not want to suffer from interference. We formulate the constraint of interference as follows. 
	
	For UE $j \in \mathcal J$, we have  multipath time-varying channel $H_{t,j}$. The noise-free received signal strength is $y_{t,j} = |H_{t,j} x_t s|$, which is required to be less than a threshold $\xi_j$. This constraint aims at reducing interference \cite{wang2019interference} and achieving a large SINR for all UEs in $\mathcal J$. We define the constraint for UE $j$ as $c_{t,j} =  y_{t,j}  -\xi_j \leq 0 , \forall j \in \mathcal J$. Combining $|\mathcal J|$ constraints through defining the maximum cost as a fixed function 
	\begin{align*}
		g_t(x_t) = \max_{j \in \mathcal J} \{c_{t,j} \}.
	\end{align*}
	We further assume that the observation is noisy with Gaussian noise $\epsilon_t$, then, we have the interference constraint as follows
	\begin{align} \label{equ: constraint}
		c_t =  g_t(x_t) + \epsilon_t \leq 0.
	\end{align}
	
	In summary, we want to control the beamforming vector $x_t$, to find the beam which achieves largest RSS $f_t(x_t)$ for UE A while making RSS for UEs in $\mathcal J$ less than a threshold, i.e., constraint $g_t(x_t) \leq 0$. Both $f_t$ and $g_t$ are time-varying unknown functions and the observations $r_t$ and $c_t$ are noisy such that $\mathbb{E}[r_t] = f_t(x_t)$ and $\mathbb{E}[c_t] = g_t(x_t)$. In heterogeneous environments with multiple transmitters, this setting is common as $\mathcal J$ includes UEs in the same cell in which A resides and UEs in the neighboring cells.
	
	%Although we consider the interference constraint, the formulation in  \eqref{equ: constraint} can be used for other constraints. For example, one constraint can be that the transmit power of the selected beam should not exceed a threshold. %\yuntian{Remove it or not?}
	
	%\subsection{Problem Formulation}
	With different beamforming vectors $x_t$, the transmitter can turn the beam in various directions and widths. Our goal is to find the best beam setting $x_t$ to maximize the received energy, with time-varying interference constraints and non-stationary channel conditions. 
	Let $T$ denote the time period for beam alignment before downlink data transmission. We formulate this beam alignment into an online constrained stochastic optimization problem. Here, the algorithm picks beamforming weight sequence $\{x_t \}_{t=1}^T$ to maximize the expected received signal strength $ \mathbb{E} [r_t]$ subject to a constraint on the  violation of interference constraint $  \mathbb{E} [c_t] \leq 0$.  We thus have an optimization problem with a constraint at time $t \in [1,T]$ as follows:
	%(If the cumulative violation is high,  the quality of the connection can be quite poor for long periods of time \cite{haenggi2009interference}.)
	\begin{align} \label{equ: objetive1}
		&\max_{ x_t } f_t(x_t)  \textit{\space s.t. \space}  g_t(x_t) \leq 0. 
	\end{align}
	%\label{equ: objetive2}
	At every time $t$, we want to use the past observations $\{(r_1,c_1),\ldots,(r_{t-1},c_{t-1})\}$ to estimate the channel gains for each beam and pick an appropriate beamforming vector $x_t$ that maximizes the above optimization problem. In other words, we want to determine a learning policy $\pi_t$, which takes as input  $\{(r_1,c_1),\ldots,(r_{t-1},c_{t-1})\}$ and outputs $x_t$, without knowing time-varying functions $f_t$ and $g_t$.

	%%%%%%%%%%%%%%%%%%%%%%%%
	
	\section{Constrained Non-stationary Gaussian Process Bandit Model}
	%This beam alignment formulation fits well into the non-stationary kernelized bandit problem \cite{chowdhury2017kernelized} with constraints. We introduce it formally as follows. 
	
	\subsection{Preliminaries}

	\textbf{Regularity Assumptions}: As similar beamforming vectors will produce similar beams,  the beamforming vector $x_t$ is highly correlated. We use the Reproducing Kernel Hilbert Space (RKHS) induced by a kernel to model this correlation.  This formulation is different from the traditional multi-armed bandit, where each beamforming vector is modeled as orthogonal arms and the correlation information between them can not be modeled.
	We assume that $f_t$ is a fixed function in a RKHS with a bounded norm.  The RKHS, denoted by $\mathbb H_k(X)$, is completely specified by its kernel function $k(\cdot, \cdot)$, with an inner product $\left< \cdot, \cdot \right>_{\mathbb H_k}$ satisfying the reproducing property: $f_t(x) = \left <f_t, k(x, \cdot) \right >_{\mathbb H_k}$ for all $f_t \in \mathbb H_k(X)$. Similar argument holds for constraint $g_t$ with kernel $\tilde k$, and $g_t \in \mathbb H_{\tilde k}(X)$.
	%Specifically,  $X$ is compact since the number of beams $|X|$ is finite. 
	%/independent
	\begin{assumption}[Boundedness]
		We assume that $f_t$ at each time $t$ is bounded by $\|f_t\|_{\mathbb H_k} \leq B$ and $k(x,x) \leq 1$ for a fixed constant $B$. Similarly, we assume that $||g||_{\mathbb H_{\tilde k}} \leq G$ and $\tilde k(x,x) \leq 1$ with a constant $G$.
	\end{assumption}
	
	Different from the restrictive assumption on unimodal reward function in \cite{hashemi2018efficient}, our assumption supports both single-path and multi-path channel and various channel models.  This is because an arbitrary continuous function can be approximated by an element in this RKHS under the supremum norm\cite{micchelli2006universal}. Moreover, it holds for practically relevant kernels. One concrete example is Squared Exponential kernel, defined as $k_{SE}(x,x') = \text{exp}(-s^2/2l^2)$
	where scale parameter $l>0$ and $s=\|x-x'\|_2$ specifies distance between two points. In the following assumption, we assume that there exists one beam satisfies the constraint in \eqref{equ: constraint}.
	\begin{assumption}[Slater Condition]
		\label{assump: slater}
		There exists a constant $\tau >0$ such that for any $t$,  there exists $\pi^o_t$ such that $\mathbb{E}_{\pi^o_t} g_t(x) \leq - \tau$.%, where  $\pi^o_t $ is a policy over the set of $X$.
	\end{assumption}	
	This is a mild assumption since it only requires the existence of a policy at time $t$ such that the constraint is less than a strictly negative value. From interference  perspective, it is satisfied if one beamforming vector $x_t$ has $g_t(x_t) \leq -\tau$.

	\textbf{Time-varying Budget}: As the channel condition under mobility is time-varying, we assume that the total variation between $f_t$ and $f_{t+1}$ satisfies the following budget, which includes both abruptly-changing and slowly-changing environments. As the channel for UEs in $\mathcal J$ is also non-stationary, a similar budget holds for the constraint $g_t$ as well.

	\begin{assumption}[Varying Budget]
		We assume that the variations in  reward and constraint functions are bounded, i.e. $\sum_{t=1}^{T-1} \|f_{t+1}-f_t\|_{\mathbb H_k} \leq B_f$ and $\sum_{t=1}^{T-1} \|g_{t+1}-g_t\|_{\mathbb H_{\tilde k}} \leq B_g$.
	\end{assumption}
	
	The variation budget will show up in the upper bounds of regret and constraint violation. In the next subsection, we develop Algorithm \ref{alg}  without prior knowledge of this budget. In Corollary \ref{cor: bound} of Section V, we show that the algorithm will have a better performance guarantee if we know this budget.
	
	%Without this assumption, we cannot provide a reasonable bound for the performance of any algorithm, as functions $f_t$ and $g_t$ could change dramatically.

	\textbf{Gaussian Process agnostic setting}: We recall the surrogate model in standard GP-UCB algorithm (Kernelized bandit) \cite{chowdhury2017kernelized}.  Gaussian process (GP) and Gaussian likelihood models are used to design this algorithm only. $GP(0,k(\cdot, \cdot ))$, a Gaussian process with zero mean and covariance $k$, is the prior for the reward function $f_t$. The noise $n_t$ is drawn independently from $\mathcal N(0,\lambda)$. Conditioned on the history $\mathcal H_t$, it has the posterior distribution of $f_t$ as $GP \left(\mu_t(\cdot),\sigma_t^2(\cdot) \right)$, where the posterior mean and variance are defined as follows.
	\begin{align}
		\mu_t(x)  &= k_t(x)^T (K_t + \lambda I )^{-1} r_{1:t}. \label{equ: mean}\\
		\sigma^2_t(x) &= k(x,x) -k_t(x)^T (K_t + \lambda I )^{-1} k_t(x). \label{equ: var}
	\end{align} 
	where $r_{1:t}\in \Re^t$ is the reward vector $[r_1,\ldots,r_t]^T$. For set of sampling points $A_t=\{x_1, \ldots, x_t\}$, the kernel matrix is $K_t = [k(x,x')]_{x,x' \in A_t} \in \Re^{t \times t}$ for kernel function $k$ and the vector $k_t(x)=[k(x_1,x), \ldots, k(x_t,x)]^T \in \Re^t$.

	Similarly, for constraint $g_t$, we have an associated posterior mean $\tilde \mu_t$ and posterior variance $\tilde \sigma^2_t$, where $r_{1:t}$ is replaced by $c_{1:t}$ and kernel $k$ is replaced by $\tilde k$.
	The GP prior and Gaussian likelihood are only used for algorithm design and do not affect the setting of reward function $f_t \in \mathbb H_k(X)$, constraint function $g_t \in \mathbb H_{\tilde k}(X)$, and sub-Gaussian noise $n_t$, $\epsilon_t$. We further define  the maximum information gain \cite{srinivas2009gaussian} as follows. 
	\begin{definition} For a given kernel $k(x,x')$, the maximum information gain at time $t$ is 
		$\gamma_t := \max_{A\subset X: |A|=t} I(y_A;f_A)  = \max_{A\subset X: |A|=t}  \frac{1}{2} \log \det(I_t + \lambda^{-1} K_A)$, where $K_A= [k(x,x')]_{x,x'\in A}$ and $I_t$ is $t \times t$ identity matrix. 
	\end{definition}
	In the definition, $I(y_A;f_A)$  denotes the mutual information between $f_A=[f(x)]_{x\in A}$ and $y_A = f_A + n_A$, which quantifies the reduction in uncertainty about $f$ after observing $y_A$ at points $A$.  We have $\gamma_t = O(d \log t)$ for linear kernel, and  $\gamma_t = O( (\log t)^{d+1})$ for Squared Exponential kernel \cite{srinivas2009gaussian}.

	\subsection{Algorithm}
	\begin{algorithm}[t] 
		\SetAlgoLined
		\SetKwInOut{Input}{Input}\SetKwInOut{Initialization}{Initialization}
		\Input{Kernel $k(\cdot, \cdot), \tilde k(\cdot, \cdot), \gamma_t,  \lambda,\delta$, restart interval $W$, $\rho = \frac{4B}{\tau}$, $\eta=\frac{\rho}{G \sqrt{T}}$ }, 
		%\BlankLine
		\For{$t\geq 1$}{
			\If{t mod W = 1}{
				reset to the initialization state, $t_0= t$}
			Set $ \beta_{t-1} =   B+  \frac{1}{\sqrt{\lambda}}  R \sqrt{  2\log(\frac{1}{\delta}) +  2  \gamma_{t-t_0}}$,  $ \tilde \beta_{t-1} =   G+  \frac{1}{\sqrt{\lambda}}  R \sqrt{  2\log(\frac{1}{\delta}) +  2  \tilde \gamma_{t-t_0}}$ \;
			Let $ \hat f_t(x)=Proj_{[-B,B]} \mu_{t-1} (x) + \beta_{t-1} \sigma_{t-1} (x)$, and  $\hat g_t(x)= Proj_{[-G,G]} \tilde \mu_{t-1} (x) - \tilde \beta_{t-1} \tilde \sigma_{t-1} (x)$\;
			Define acquisition: $\hat z_{\phi_t}(x) =  \hat f_t(x)  - \phi_t \textit{\space}  \hat g_t(x)$\;
			Choose beamform vector $x_t = \arg \max_{x \in X} \hat z_{\phi_t}(x)$,  Observe received signal $r_t = f_t(x_t) + n_t$, observe interference constraint $c_t = g_t(x_t) + \epsilon_t$\;
			Update $\mu_t(x)$, $\sigma_t(x)$ and $\tilde \mu_t(x)$, $\tilde \sigma_t(x)$ in \eqref{equ: mean} and \eqref{equ: var} with data $(x_t, r_t, c_t)$ from $t_0$ to $t$\;
			Update dual $\phi_{t+1}=Proj_{[0,\rho]} [\phi_t  + \eta  \hat g_t(x_t) ]$.
		}
		\caption{Restart GP-UCB with Constraints} {\label{alg}}
	\end{algorithm}
	
	We devise here a learning policy $\pi_t$ that employs restarting for dealing with learning in time-varying environments and a primal-dual method for solving the constrained optimization. The restart schedule is shown in lines 1 to 3 in Algorithm \ref{alg}. After $W$ time slots, we will discard previous estimates and restart to build a new estimate. The actual value of $W$  will be discussed in the next section. 
	
	Second, we construct an optimistic estimate on received signal $f_t$ and interference constraint $g_t$ based on  Upper Confidence Bound (UCB) bandit.  Specifically, the UCB-type exploration combines mean and variance  $\mu_{t-1} (x) + \beta_{t-1} \sigma_{t-1} (x)$ for $f_t(x)$, and $ \tilde \mu_{t-1} (x) - \tilde \beta_{t-1} \tilde \sigma_{t-1} (x)$ for $g_t(x)$ towards negative constraint in line 6. These estimates will be truncated according to their bounds $B$ and $G$ respectively.

	Third, we employ the primal-dual optimization to solve the constrained optimization problem.  Let the baseline problem  at time $t$ be given by $\max_{\pi_t} \{ \mathbb{E}_{\pi_t} f_t(x):  \mathbb{E}_{\pi_t} g_t(x) \leq 0 \}$. Then the associated Lagrangian is $L(\pi_t, \phi_t) =  \mathbb{E}_{\pi_t} [f_t(x)] - \phi_t \mathbb{E}_{\pi_t} [g_t(x)] $ and the dual problem is $D(\phi_t)=  \max_{\pi_t}   L(\pi_t, \phi_t) $.  If we approximate $ \mathbb{E}_{\pi_t} [f_t(x)]$ with $\hat f_t(x)$ and $\mathbb{E}_{\pi_t} [g_t(x)] $ with $\hat g_t(x)$, then acquisition $\hat z_{\phi_t}(x) =  \hat f_t(x)  - \phi_t  \hat g_t(x)$\ is close to $ L(\pi_t, \phi_t)$. Therefore, we can get the proximate optimal solution to $D(\phi_t)$ as $x_t = \arg \max_{x \in X} \hat z_{\phi_t}(x)$, as shown in line 7 and 8 in Algorithm \ref{alg}.

	Fourth, we update the posterior mean and variance via standard Gaussian Process regression in \eqref{equ: mean} and \eqref{equ: var}, which works for both $f_t$ and $g_t$ with kernel $k$ and $\tilde k$ respectively.
	
	Finally, we design a dual variable update towards minimizing $D(\phi_t)$. As the gradient of $\hat z_{\phi_t}(x)$ (approximating $D(\phi_t)$) with respect to $\phi_t$ is $-\hat g_t(x)$, we take a projected gradient descent with step size $\eta$. The projected upper bound is selected to be $\rho \geq \frac{4B}{\tau}$, which utilizes the information that the optimal dual $\phi^*$ is within range $[0, \frac{2B}{\tau}]$. The final dual update becomes $\phi_{t+1}=Proj_{[0,\rho]} [\phi_t  + \eta  \hat g_t(x_t) ]$.

	Except for using the restart strategy \cite{zhou2021no, zhao2020simple} to leverage the time-varying environment, other forgetting methods will also work within our framework, such as the sliding window \cite{zhou2021no, cheung2019learning} and the exponentially increasing weights \cite{deng2022weighted, russac2019weighted}.

	%%%%%%%%%%%%%%%%%%%%%%%%%%%%%%%%%%%%%%%%%%%%%%%
	\section{Performance Analysis}

	%\subsection{Regret and constraint violation}
	In this section, we assess the performance of Algorithm~\ref{alg} and provide theoretical bounds for both  signal strength and interference constraint. 
	In the online and non-stationary setting, it is difficult to directly  solve the problem defined in  \eqref{equ: objetive1}  by choosing the best policy at time $t$. This is because both the reward and constraint function are time-varying and the noisy observation is revealed only after a policy is executed. In addition, since the constraint function $g_t$ is unknown, the constraint $g_t(x_t)\leq 0$ may not be satisfied in every time $t$. Instead, we aim to satisfy the long-term constraint $\sum_{t=1}^T g_t(x_t)\leq 0$ over a given period $T$ \cite{zhou2022kernelized}. We define the  \emph{dynamic regret} $R(T)$ as the difference of total reward between our policy $\{x_t \}_{t=1}^T$ and the best policy $\{\pi_t^* \}_{t=1}^T$. We also define the \emph{constraint violation} $V(T)$  as the cumulative violation of constraint over period $T$. We will show that the proposed algorithm  achieves both small $R(T)$ and $V(T)$.
	%To assess  the performance of our algorithm,
	\begin{align*}
		R(T) &=  \sum_{t=1}^T \mathbb{E}_{\pi^*_t} f_t(x) - \sum_{t=1}^T f_t(x_t). \\
		V(T) &= \left[ \sum_{t=1}^T g_t(x_t) \right]_{+}.
	\end{align*}
	where  $[\cdot]_{+}:= \max \{ \cdot, 0\}$. $\pi^*_t$ is defined as the best policy at time $t$ which maximizes $\mathbb{E}_{\pi_t} f_t(x) = \int_{x \in X} f_t(x) \pi_t(x) dx$ while  satisfying the constraint $\mathbb{E}_{\pi_t} g_t(x) = \int_{x \in X} g_t(x) \pi_t(x)  dx $ $\leq 0$. We highlight that the optimal policy $\pi^*_t$ is time-dependent and may change over time, as $f_t$ and $g_t$ change over time. 
	% (\xingyu{definition for this?})
	
	%\subsection{Sublinear Bounds}
	
	In the following theorem, we obtain upper bounds on the regret of the received signal strength $R(T)$ and the interference constraint violation $V(T)$, both of which are on the order of $T^{\frac{3}{4}}$. Combining the kernel of the reward function $f_t$ and the constraint function $g_t$, we further define combined maximum information gain $\hat \gamma_T = \max\{ \gamma_T, \tilde \gamma_T \} $ and combined variation budget $B_{\Delta}= \max \{B_g, B_f \}$. We first state the result without any prior knowledge of the combined variation budget $B_{\Delta}$. %i.e., Algorithm \ref{alg} has no prior knowledge of  the total variation of channel gain. 
	
	%$\hat \beta_T = \max\{ \beta_T, \tilde \beta_T \} = O(\hat \gamma_T^{\frac{1}{2}})$,
	\begin{theorem} \label{thm: bound}
		If the restart period $W= \hat \gamma_T^{\frac{1}{4}} T^{\frac{1}{2}}$, for any $\rho \geq \frac{4B}{\tau}$, the regret $R(T)$ and constraint violation  $V(T)$  in Algorithm~ \ref{alg} at time $T$ are bounded by:
		\begin{align*}
			R(T) &\leq  O \left(\rho G  \sqrt{T} + \rho \hat \gamma_T^{\frac{7}{8}} B_{\Delta} T^{\frac{3}{4}}\right).\\
			V(T) &\leq O \left(\left(1+\frac{1}{\rho}\right) \hat \gamma_T^{\frac{7}{8}} B_{\Delta} T^{\frac{3}{4}} + G\sqrt{T} \right).
		\end{align*}
	\end{theorem}
	\begin{proof}
		The complete proof is provided in Section \ref{sec: proof}.
	\end{proof}

	If the combined variation budget $B_{\Delta}$ is known,  the order of $B_{\Delta}$ in the upper bounds can be reduced from $1$ to $\frac{1}{4}$, with a different value of restart period $W$.

	\begin{corollary} \label{cor: bound}
		If restart period $W= \hat \gamma_T^{\frac{1}{4}} T^{\frac{1}{2}} B_{\Delta}^{-\frac{1}{2}}$, we have tighter upper bounds for both regret and constraint violation.
		\begin{align*}
			R(T) &\leq  O  \left(\rho G  \sqrt{T} + \rho \hat \gamma_T^{\frac{7}{8}} B_{\Delta}^{\frac{1}{4}} T^{\frac{3}{4}}\right).\\
			V(T) &\leq O  \left( \left(1+\frac{1}{\rho}\right) \hat \gamma_T^{\frac{7}{8}} B_{\Delta}^{\frac{1}{4}} T^{\frac{3}{4}} + G\sqrt{T} \right).
		\end{align*}
	\end{corollary}
	
	\begin{remark}
		Our dynamic regret bounds match the order of the regret in the nonstationary Gaussian Process bandits without any constraints \cite{deng2022weighted,zhou2021no,zhao2021non}. It also generalizes the kernelized bandits with constraints \cite{zhou2022kernelized} in stationary environments to the time-varying case.
	\end{remark}

	%%%%%%%%%%%%%%%%%%%%%%%%%%%%%%%%%%%%%%%%%%%%%%
	\section{Numerical Results}
	
	\begin{figure*}
		\centering
		\begin{subfigure}{.33\textwidth}
			\centering
			\includegraphics[width=6cm, height=4cm]{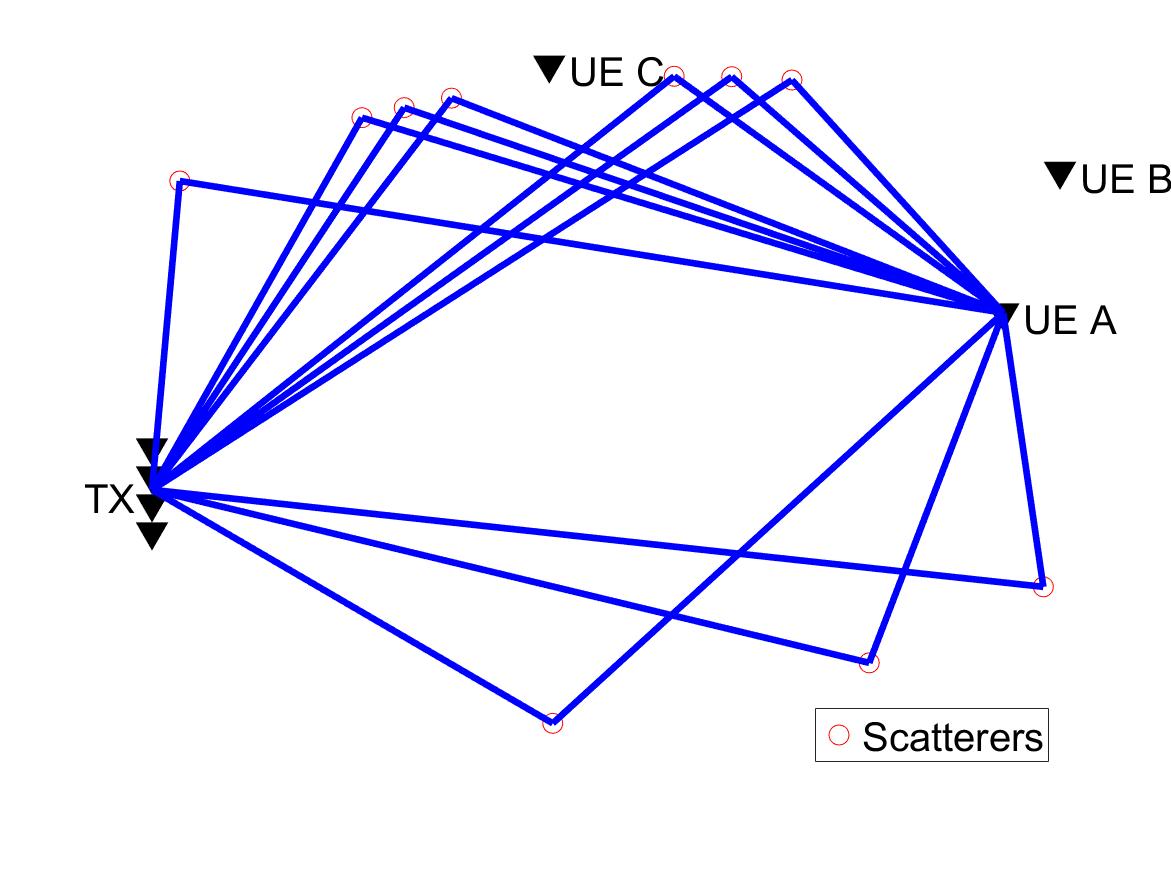}
			\caption{Channel with 10 NLOS paths}
		\end{subfigure}%
		%  \hfill
		\begin{subfigure}{.33\textwidth}
			\centering
			\includegraphics[width=6cm, height=4cm]{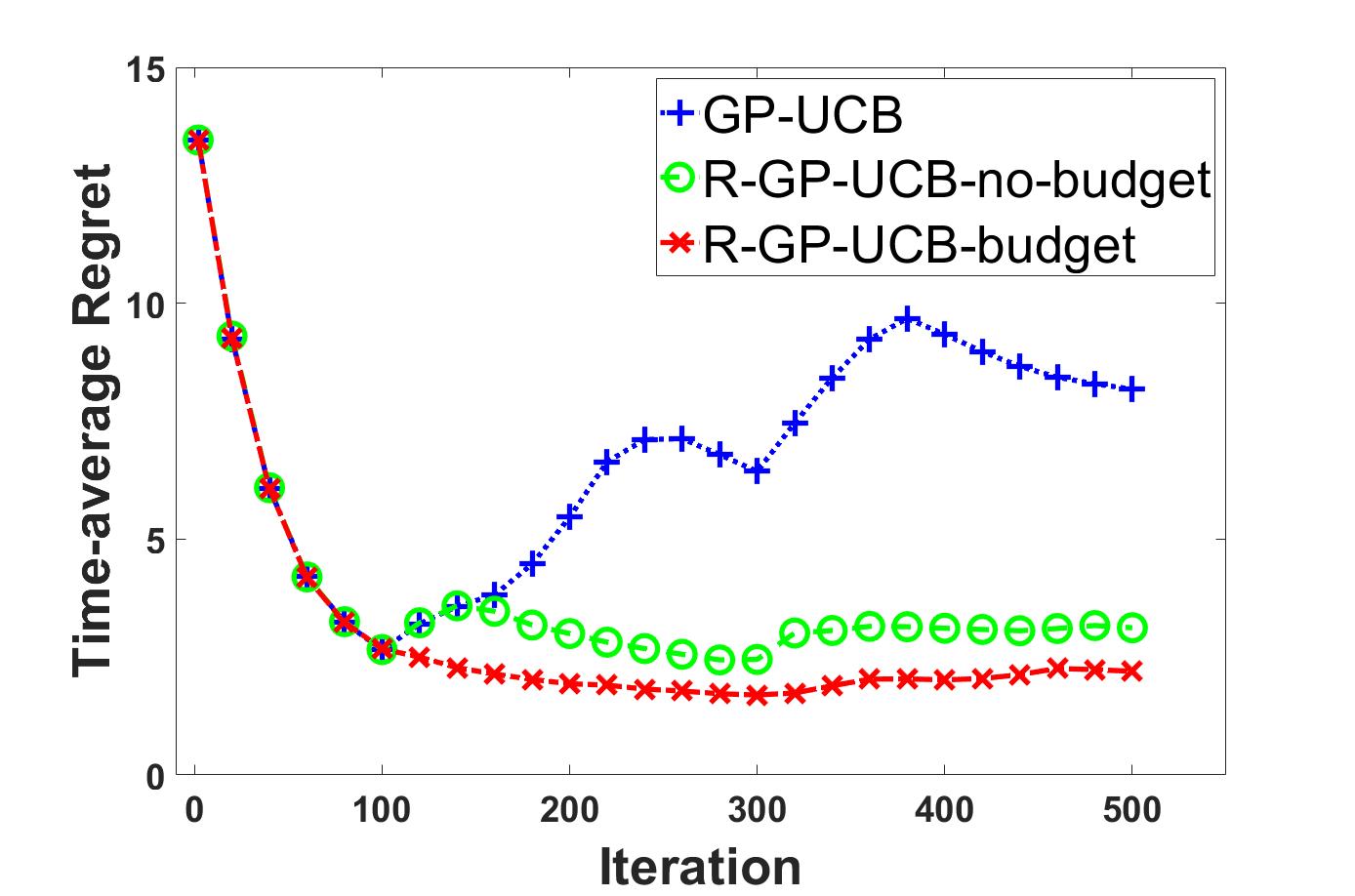}
			\caption{Time-average Regret, Abruptly-change}
		\end{subfigure}%
		\begin{subfigure}{.33\textwidth}
			\centering
			\includegraphics[width=6cm, height=4cm]{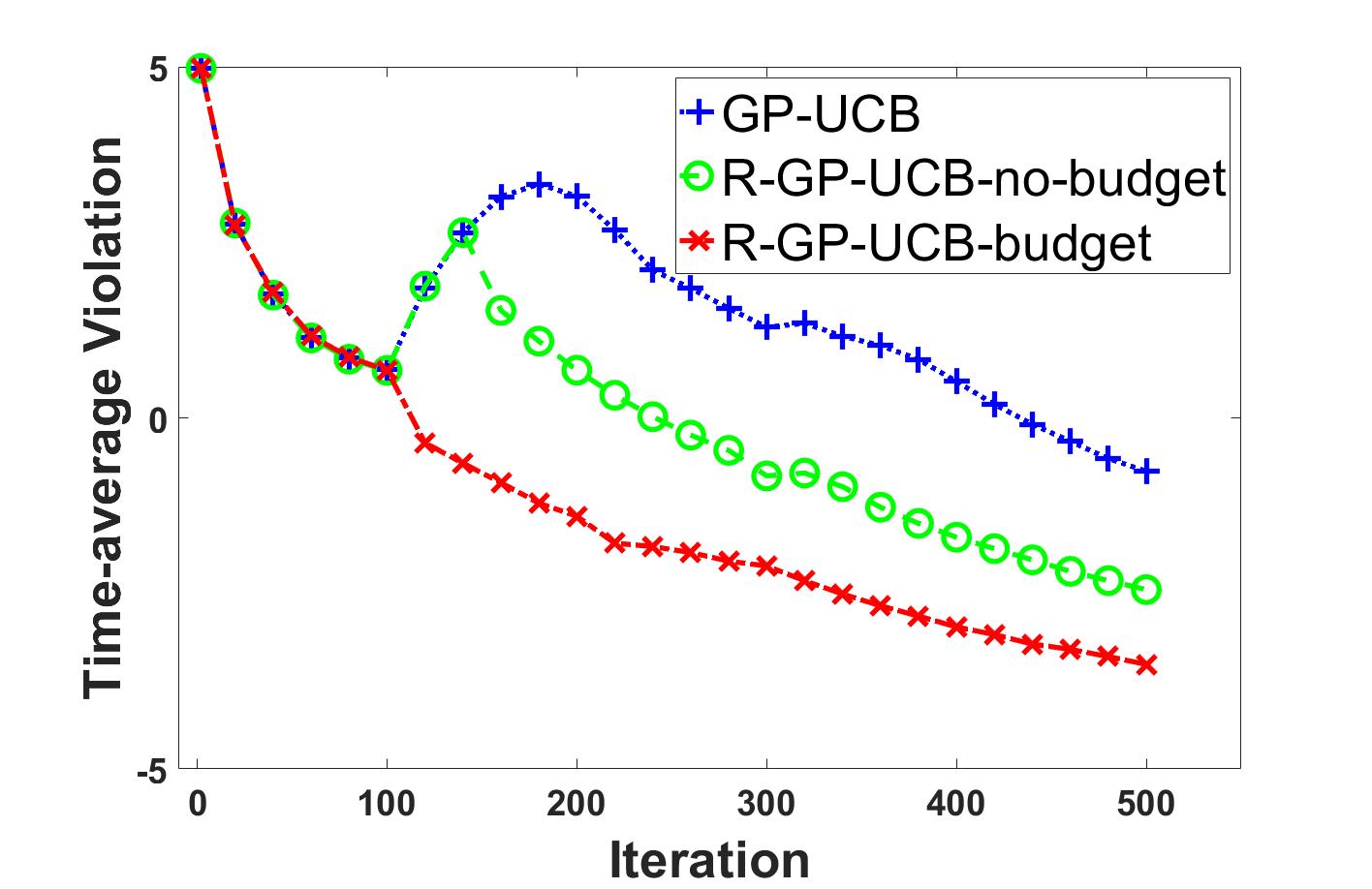}
			\caption{Time-average Violation, Abruptly-change}
		\end{subfigure}%
		\caption{Time-average Performance of two algorithms under abruptly change at time 100 and 300}
		\label{fig: simulation1}
	\end{figure*}
	
	\begin{figure*}
		\centering
		\begin{subfigure}{.5\textwidth}
			\centering
			\includegraphics[width=6cm, height=4cm]{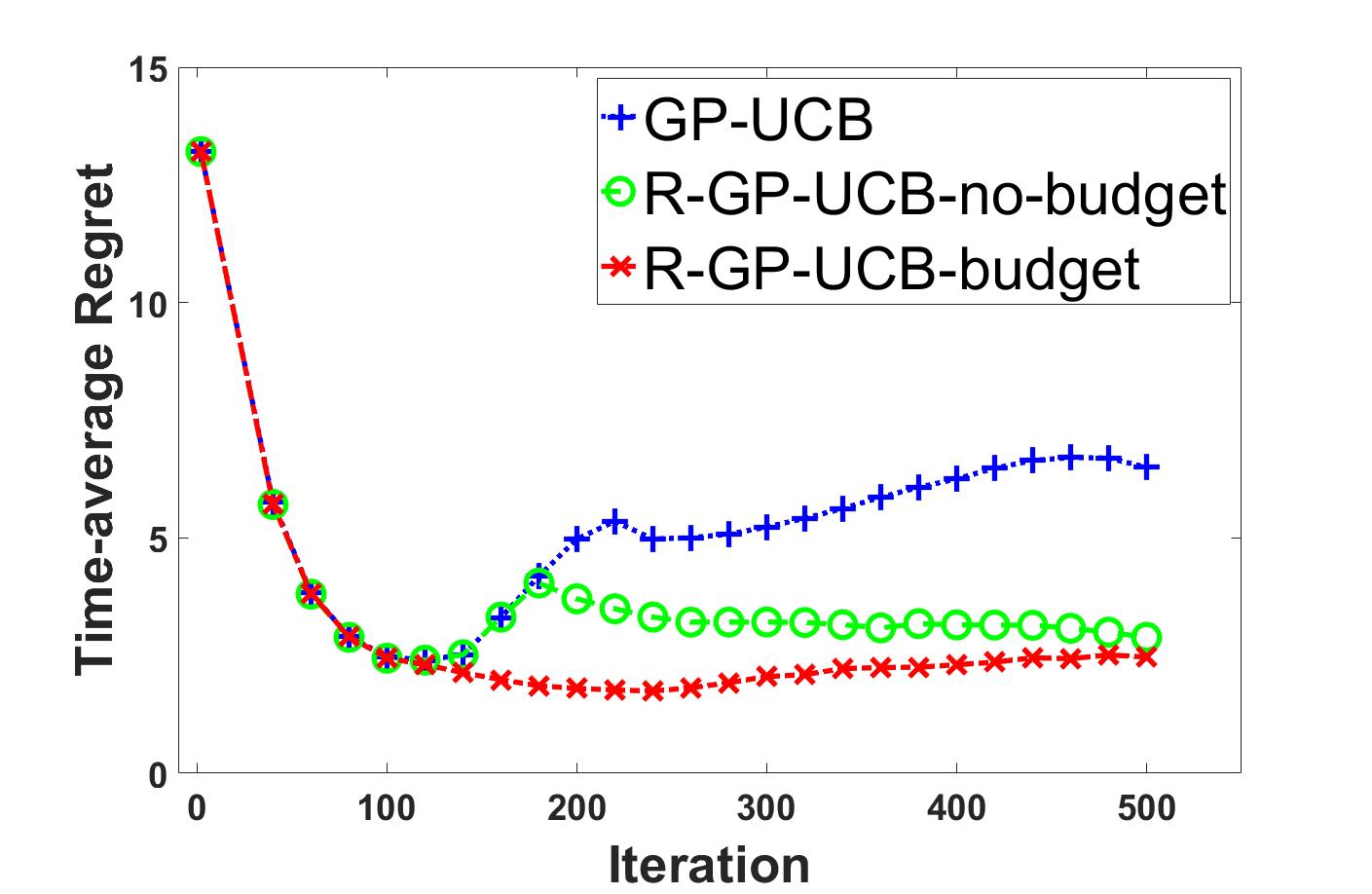}
			\caption{Time-average Regret, Slowly change}
		\end{subfigure}%
		%  \hfill
		\begin{subfigure}{.5\textwidth}
			\centering
			\includegraphics[width=6cm, height=4cm]{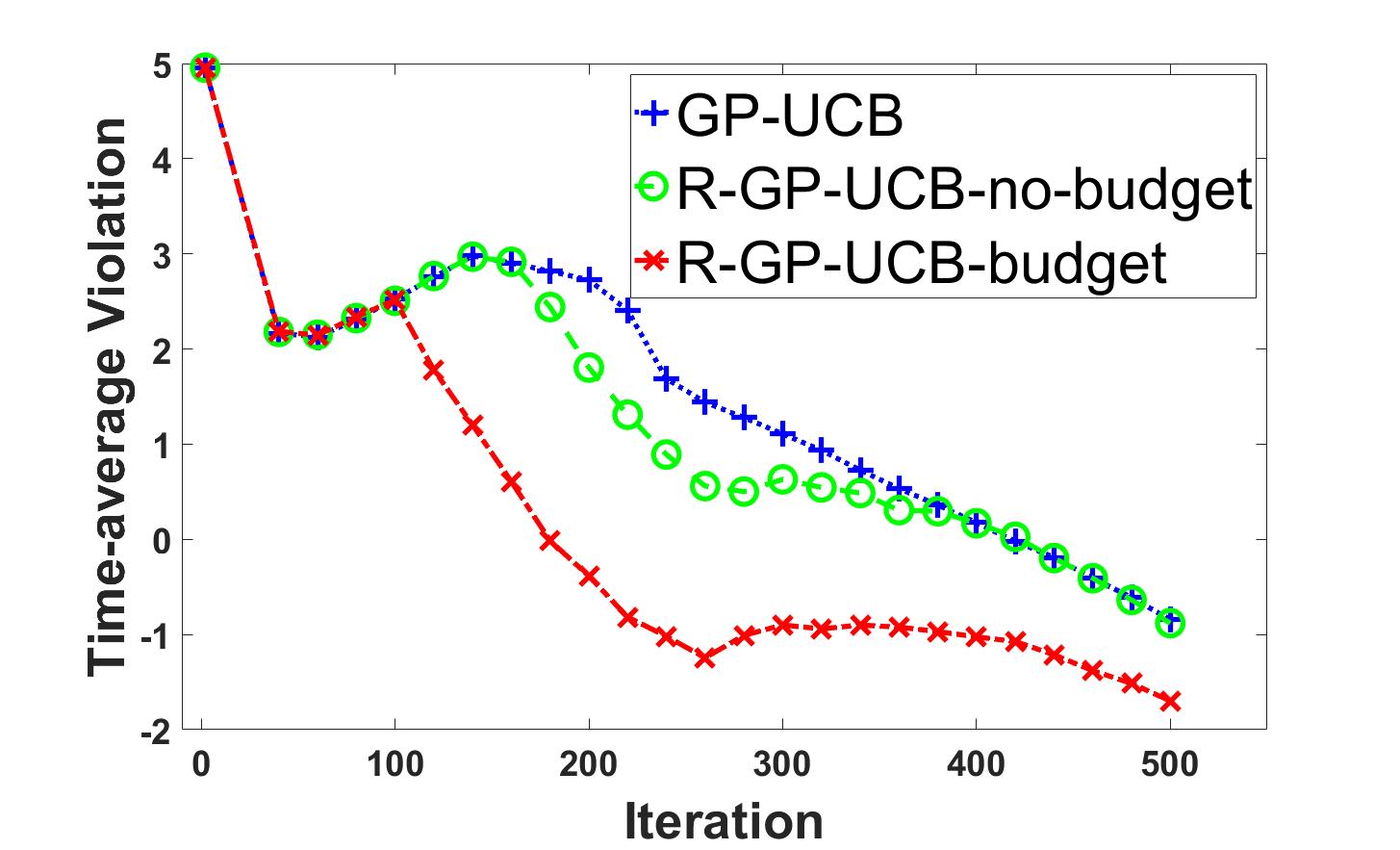}
			\caption{Time-average Violation, Slowly change}
		\end{subfigure}%
		\caption{Time-average Performance of two algorithms under slowly change scenarios}
		\label{fig: simulation2}
	\end{figure*}

	We compare the performance of the proposed algorithm with GP-UCB algorithm with Constraints (i.e. CKB-UCB algorithm in \cite{zhou2022kernelized}). We perform the comparison under both abruptly-changing environments and the slowly-varying scenarios. We use the Phased Array System Toolbox in Matlab to simulate the beamforming vector $x_t$ and channel $H_t$. As shown in Figure  \ref{fig: simulation1} (a), we assume the transmitter has 4 antennas in the uniform linear array. The channel consists of 10 NLOS path and the carrier frequency is 60GHz. We evaluate 100 beam candidates and use Squared Exponential kernel with scale parameter $l=1$ to capture the correlation between $x_t$. The received signal strength function $f$ and interference constraint function $g$ are achieved through Monte Carlo simulation, where path gains change due to moving of UEs.
	
	The first experiment is conducted in the abruptly-changing environments, where channel changes happen at time slot $100$ and $300$. As shown in Figure \ref{fig: simulation1}, GP-UCB has a significant increase in the time-average regret because of changes in channel gains. Our proposed R-GP-UCB overcomes it whenever the variation budget $B_{\Delta}$ is known or not. The time-average interference violations increases at time $100$, which means the constraint is violated. After some time slots, it keeps decreasing, which indicates that the adaptively chosen beamforming vector $x_t$ satisfies the constraint $g_t(x_t)\leq 0$ on average. The second experiment corresponds to a slowly-varying scenario, where all UEs keep moving at a random direction at each time slot. The received signal $f_t$ and the constraint $g_t$ are simulated per time slot and keep changing.

	%%%%%%%%%%%%%%%%%%%%%%%%%%%%%%%%%%%%%%%%%%%%%%%%%%
	%%%%%%%%%%%%%%%%%%%%%%%%%%%%%%%%%%%%%%%%%%%%%%%%%%%
	\if 0
	
	\begin{figure*}
		\centering
		\begin{subfigure}{.25\textwidth}
			\centering
			\includegraphics[width=4.6cm, height=3cm]{new_figures/regret_abruptly_change.jpg}
			\caption{Time-average Regret, Abruptly-change}
		\end{subfigure}%
		%  \hfill
		\begin{subfigure}{.25\textwidth}
			\centering
			\includegraphics[width=4.6cm, height=3cm]{new_figures/violation_abruptly_change.jpg}
			\caption{Time-average Violation, Abruptly-change}
		\end{subfigure}%
		%\caption{Time-average Performance under abruptly-changing environments where the change happens at time 100 and 300.}
		%\label{fig: simulation1}
		%\end{figure*}
		%\begin{figure*}
		%\centering
		\begin{subfigure}{.25\textwidth}
			\centering
			\includegraphics[width=4.6cm, height=3cm]{new_figures/regret_slowly_change.jpg}
			\caption{Time-average Regret, Slowly change}
		\end{subfigure}%
		%  \hfill
		\begin{subfigure}{.25\textwidth}
			\centering
			\includegraphics[width=4.6cm, height=3cm]{new_figures/violation_slowly_change.jpg}
			\caption{Time-average Violation, Slowly change}
		\end{subfigure}%
		\caption{Time-average Performance of algorithms under abruptly-changing environments and slowly-varying scenarios.}
		\label{fig: simulation}
	\end{figure*}
	
	\fi
	%%%%%%%%%%%%%%%%%%%%%%%%%%%%%%%%%%%%%%%%%%%%%
	%%%%%%%%%%%%%%%%%%%%%%%%%%%%%%%%%%%%%%%%%%%%%%

	\section{Proofs of Theorems} \label{sec: proof}
	
	%\xingyu{As I have emphasized many times...Always first highlight the challenges before the details by saying something like: Our main results are obtained via novel technical contributions. In particular, a direct combination of techniques from constrained bandits and non-stationary bandits would fail. This is because... To this end, we develop a new... This technique could also be of independent interest for... }
	
	Our main results are obtained via novel technical contributions. In particular, a direct combination of techniques from constrained bandits~\cite{zhou2022kernelized}  and non-stationary bandits~\cite{zhou2021no} would fail. This is because the standard `time-average' trick \cite{ding2021provably} does not hold when the unknown functions are time-varying. To this end, we introduce a novel upper bound on the constraint violation \emph{per time slot}. This new per-time-slot bound can be either positive or negative, while previous time-average bound is always positive. This technique could also be of independent interest for constrained reinforcement learning. Specifically, we state a proof of Theorem \ref{thm: bound} and Corollary~\ref{cor: bound} respectively. The key idea is to construct a decomposition in Lemma \ref{lemma: no_xi_bound_t} via dual variable update and concentration inequality. This decomposition yields a bound on $R(T)$ by choosing $\phi=0$ and a bound on $V(T)$ by choosing $\phi=\rho$. The novel upper bound on the constraint violation per time slot is  Lemma  \ref{lem: new gt}.
	%\xingyu{give reference here} 
	
	%The existing result in \cite{zhou2022kernelized} fails since reward and constraint functions are time-varying and the time-average result does not hold. We overcome it by introducing a novel upper bound on the constraint violation per time slot, i.e. Lemma  \ref{lem: new gt}. This new per-time-slot bound can be either positive or negative, while previous time-average bound is always positive.

	%Compared with existing results \cite{zhou2022kernelized}, we develop a new bound on $g_t(x_t)$ in Lemma  \ref{lem: new gt} for any time $t$ rather than its time-average. This new bound can be either positive or negative, while previous bounds are always positive.
	
	%holds for any $g_t(x_t)$, while previous results only hold if $g_t(x_t)\geq 0$. 

	\subsection{Proof of the Regret Bound $R(T)$}
	
	In this subsection, we focus on establishing the regret bound. We first state the concentration inequality  for non-stationary UCB algorithms, which bounds the gap between estimator $\mu_{t-1}(x)$ and true reward value $f_t(x)$. In comparison to the stationary case, the non-stationary UCB algorithm introduces an extra term $\Delta f_t$. A similar result holds for the constraint function $g_t$, where kernel $k$ is replaced by kernel $\tilde k$.
	
	\begin{lemma} {\label{lemma: fbound}}
		We have the following bound for any $x$:
		\begin{align*}
			&|\mu_{t-1} (x) - f_t(x)| \leq \Delta f_t  + \beta_t \sigma_{t-1}(x),  \\
			\textit{where }& \Delta f_t := 1/\lambda \sqrt{2 (1+\lambda) W \gamma_W} \sum_{s=t_0}^{t-1} ||f_s - f_{s+1} ||_{\mathbb H_k}.
		\end{align*}
		Similarly, for $g_t$ we have $|\tilde \mu_{t-1} (x) - g_t(x)| \leq \Delta g_t  + \tilde \beta_t \tilde \sigma_{t-1}(x).$ and $ \Delta g_t := 1/\lambda \sqrt{2 (1+\lambda) W \tilde \gamma_W} \sum_{s=t_0}^{t-1} ||g_s - g_{s+1} ||_{\mathbb H_{\tilde k}}.$
	\end{lemma}
	\begin{proof}
		This lemma is the same  as Lemma 1 in \cite{zhou2021no}. 
	\end{proof}

	Leverage the maximizing action at the  acquisition function $\hat z_{\phi_t} (x)$, we have the following result.

	\begin{lemma} \label{lemma: Epi_ft} 
		The estimator $\hat f_t(x_t)$ is bounded as 
		\begin{align*}
			\mathbb{E}_{\pi^*_t} f_t(x) - \hat f_t(x_t) \leq -\phi_t \hat g_t(x_t) + \phi_t \Delta g_t + \Delta f_t. 
		\end{align*}
	\end{lemma}
	
	\begin{proof}
		As $x_t = \arg \max_{x \in X} \hat z_{\phi_t}(x)$, from the definition of $\hat z_{\phi_t}(x)$, we have $    \mathbb{E}_{\pi^*_t} \hat f_t(x) - \phi_t \mathbb{E}_{\pi^*_t} \hat  g_t(x) \leq \hat f_t(x_t) - \phi_t \hat g_t(x_t).$
		%\begin{align*} 
		%    \mathbb{E}_{\pi^*_t} \hat f_t(x) - \phi_t \mathbb{E}_{\pi^*_t} \hat  g_t(x) \leq \hat f_t(x_t) - \phi_t \hat g_t(x_t).
		%\end{align*}
		From Lemma \ref{lemma: fbound} and $\mathbb{E}_{\pi^*_t} g_t(x) \leq 0$, we have 
		\begin{align*}
			\mathbb{E}_{\pi^*_t} f_t(x) &\leq \mathbb{E}_{\pi^*_t} \hat  f_t(x) + \Delta f_t,\\
			\mathbb{E}_{\pi^*_t} \hat g_t(x) &\leq  \mathbb{E}_{\pi^*_t} g_t(x)  + \Delta g_t  \leq  \Delta g_t.
		\end{align*}
		Combining above three inequalities, we have
		\begin{align*}
			\mathbb{E}_{\pi^*_t} f_t(x) &\leq \hat f_t(x_t) - \phi_t \hat g_t(x_t) + \phi_t \mathbb{E}_{\pi^*_t} \hat g_t(x) + \Delta f_t\\
			&\leq \hat f_t(x_t) - \phi_t \hat g_t(x_t) + \phi_t \Delta g_t   +\Delta f_t.
		\end{align*} 
		This concludes the above results.
	\end{proof}

	We obtain the following lemma on the dual variable $\phi_t$.
	\begin{lemma} \label{lemma:no_xi_phigt}
		The estimator $\hat g_t(x_t)$ is bounded as 
		\begin{align*}
			(\phi-\phi_t) \hat g_t(x_t)  \leq   \eta G^2+  \frac{1}{2\eta}(\phi_t -\phi)^2  - \frac{1}{2\eta}(\phi_{t+1} - \phi)^2.
		\end{align*}
	\end{lemma}
	
	\begin{proof}
		From the dual update, we have
		\begin{align*}
			(\phi_{t+1} - \phi)^2 &\leq (\phi_t + \eta \hat g_t(x_t) - \phi)^2 \\
			&\leq (\phi_t - \phi)^2 +  2 \eta \hat g_t(x_t) (\phi_t - \phi) + \eta^2 \hat g_t^2(x_t).
		\end{align*}
		
		We prove above result as $\eta >0$ and $\hat g_t(x_t) \leq G$.
	\end{proof}

	In the following, we establish an important bound on $\mathbb{E}_{\pi^*_t} f_t(x) - f_t(x_t) + \phi g_t(x_t)$, which is the sum of instantaneous  regret of $f_t$ and instantaneous violation of $g_t$ weighted by a general dual parameter $\phi$. Both the regret bound $R(T)$ and constraint violation $V(T)$ are developed from this lemma. 
	
	\begin{lemma} \label{lemma: no_xi_bound_t}
		For any $\phi \in [0, \rho]$, we have
		\begin{align*}
			& \mathbb{E}_{\pi^*_t} f_t(x) - f_t(x_t) + \phi g_t(x_t) \leq \delta (x_t, \phi), \\
			\textit{where \space} &\delta (x_t, \phi)= \rho \Delta g_t + 2 \Delta f_t + 2 \beta_t \sigma_{t-1} (x_t) + \phi \Delta g_t  \\
			+2\phi & \tilde \beta_t \tilde \sigma_{t-1}(x_t) + \eta G^2 + \frac{1}{2 \eta} (\phi_t - \phi)^2 - \frac{1}{2 \eta} (\phi_{t+1} - \phi)^2.
		\end{align*}
	\end{lemma}

	\begin{proof}
		From Lemma \ref{lemma: fbound}, Lemma \ref{lemma: Epi_ft} and $\phi_t \leq \rho$, we have
		\begin{align*}
			\mathbb{E}_{\pi^*_t} f_t(x) - f_t(x_t)  &= \mathbb{E}_{\pi^*_t} f_t(x) - \hat f_t(x_t) + \hat f_t(x_t) - f_t(x_t) \\
			\leq -\phi_t & \hat g_t(x_t) + \rho \Delta g_t + 2\Delta f_t + 2\beta_t \sigma_{t-1}(x_t).
		\end{align*}
		
		As $\phi g_t(x_t) = \phi ( g_t(x_t) - \hat g_t(x_t) )+  \phi_t \hat g_t(x_t) + ( \phi \hat g_t(x_t) - \phi_t \hat g_t(x_t) )$, use Lemma \ref{lemma: fbound} for first bracket,  and Lemma \ref{lemma:no_xi_phigt} for last bracket, we have
		\begin{align*}
			\phi g_t(x_t) \leq & \phi \Delta g_t + \phi 2 \tilde \beta_t \tilde \sigma_{t-1}(x_t) + \phi_t \hat g_t(x_t) \\
			&+ \eta G^2+  \frac{1}{2\eta}(\phi_t -\phi)^2  - \frac{1}{2\eta}(\phi_{t+1} - \phi)^2.
		\end{align*}
		We conclude by combining above two inequalities. 
	\end{proof}

	%\begin{lemma} 
	%The instantaneous regret is bounded as 
	%\begin{align*}
	%    &\mathbb{E}_{\pi_t^*} f_t(x) - \mathbb{E}_{\tilde \pi_t}f_t(x)   \\
	%    &\leq \rho \Delta g_t + 2 \Delta f_t + 2 \beta_t \sigma_{t-1} (x_t) +  \eta G^2 + \frac{1}{2 \eta} \phi_t^2 - \frac{1}{2 \eta} \phi_{t+1}^2. 
	%\end{align*}
	%\end{lemma}
	
	%\begin{proof}
	%Use the the definition on  policy $\tilde \pi_t$, $\mathbb{E}_{\tilde \pi_t}f_t(x) = f_t(x_t)$ and Lemma  \ref{lemma: no_xi_bound_t} with $\phi= 0$.
	%\end{proof}

	Finally, by choosing $\phi=0$, we can separate the instantaneous  regret of $f_t$  from the above decomposition. Then, by summarizing the instantaneous  regret of $f_t$ over time $t$, we achieve the bound on $R(T)$. Two forms of bounds are proved by setting restart period $W$ when $B_{\Delta}$ is known or not.
	\begin{lemma} \label{lemma: regret}
		The regret is bounded as 
		\begin{align*}
			R(T) &= O (\rho G  \sqrt{T} + \rho \hat \gamma_T^{\frac{7}{8}} B_{\Delta} T^{\frac{3}{4}}).
		\end{align*}
	\end{lemma}	
	
	\begin{proof}
		Use the the definition on  policy $\tilde \pi_t$, $\mathbb{E}_{\tilde \pi_t}f_t(x) = f_t(x_t)$ and Lemma  \ref{lemma: no_xi_bound_t} with $\phi= 0$, we have the instantaneous regret is bounded as 
		\begin{align*}
			&\mathbb{E}_{\pi_t^*} f_t(x) - \mathbb{E}_{\tilde \pi_t}f_t(x)   \\
			&\leq \rho \Delta g_t + 2 \Delta f_t + 2 \beta_t \sigma_{t-1} (x_t) +  \eta G^2 + \frac{1}{2 \eta} \phi_t^2 - \frac{1}{2 \eta} \phi_{t+1}^2. 
		\end{align*}

		As $\phi_0 = 0$, summarize the above inequalities over time $t$,
		\begin{align*}
			R(T) &= \sum_{t=1}^T \mathbb{E}_{\pi_t^*} f_t(x) - f_t(x_t) \\
			& \leq \sumt \rho \Delta g_t + \sumt 2 \Delta f_t + \sumt  2 \beta_t \sigma_{t-1} (x_t) + \eta G^2 T.   \\
		\end{align*}
		
		Let $\hat \gamma_T = \max\{ \gamma_T, \tilde \gamma_T \} $, $B_{\Delta}= \max \{B_g, B_f \}$, $\hat \beta_T = \max\{ \beta_T, \tilde \beta_T \} = O(\hat \gamma_T^{\frac{1}{2}})$, $W= \hat \gamma_T^{\frac{1}{4}} T^{\frac{1}{2}}$, $\eta=\frac{\rho}{G\sqrt{T}}$, we have
		\begin{align*}
			R(T) &\leq \eta G^2 T + (\rho+4) \hat \gamma_T^{\frac{7}{8}} B_{\Delta} T^{\frac{3}{4}} \\
			&\leq O (\rho G  \sqrt{T} + \rho \hat \gamma_T^{\frac{7}{8}} B_{\Delta} T^{\frac{3}{4}}).
		\end{align*}

		If $W= \hat \gamma_T^{\frac{1}{4}} T^{\frac{1}{2}} B_{\Delta}^{-\frac{1}{2}}$, we have a tighter bound $R(T)\leq O \left(\rho G  \sqrt{T} + \rho \hat \gamma_T^{\frac{7}{8}} B_{\Delta}^{\frac{1}{4}} T^{\frac{3}{4}}\right)$ \cite{zhou2021no}.
	\end{proof}

	\subsection{Proof of the Constraint Violation Bound $V(T)$}
	
	In this subsection, we focus on establishing the upper bound of constraint violation $V(T)$. The first lemma is an extension of Lemma \ref{lemma: no_xi_bound_t} by choosing $\phi= \rho$, where $ \rho \geq \frac{4B}{\tau}$ is defined as the upper bound in truncating $\phi_t$.
	
	\begin{lemma} \label{lemma: no_xi_delta}
		We construct a decomposition as
		\begin{align*}
			&\mathbb{E}_{\pi_t^*} f_t(x) - \mathbb{E}_{\tilde \pi_t}f_t(x)  + \rho \mathbb{E}_{\tilde \pi_t}g_t(x) \leq \delta (x_t, \rho), \\
			\textit{where \space} &\delta (x_t, \rho)= \rho \Delta g_t + 2 \Delta f_t + 2 \beta_t \sigma_{t-1} (x_t) + \rho \Delta g_t \\
			+ 2\rho \tilde \beta_t &\tilde \sigma_{t-1}(x_t) + \eta G^2 + \frac{1}{2 \eta} (\phi_t - \rho)^2 - \frac{1}{2 \eta} (\phi_{t+1} - \rho)^2.
		\end{align*}
		
	\end{lemma}
	
	\begin{proof}
		Use the  definition on  policy $\tilde \pi_t$, $\mathbb{E}_{\tilde \pi_t}f_t(x) = f_t(x_t)$ and $\mathbb{E}_{\tilde \pi_t}g_t(x) = g_t(x_t) = - \tilde \tau$, and Lemma  \ref{lemma: no_xi_bound_t} with $\phi= \rho$.
	\end{proof}

	The following Lemma generalizes  Lemma 9 in  \cite{ding2021provably} and Theorem 42 in \cite{efroni2020exploration}. We build a new upper bound for $g_t(x_t)$ for all time $t$, rather than  time averaged $\frac{1}{T} \sum_t g_t(x_t)$. We note that $\delta(x_t, \rho)$ may be negative, therefore we further generalize the result from positive $g_t(x_t)$ only to the case where $g_t(x_t)$ can be either positive or negative.

	\begin{lemma} \label{lem: new gt}
		If $\rho \geq 2 \phi^*_t$, for $g_t(x_t) \in [-G, G]$ we have 
		\begin{align*}
			g_t(x_t)  = \mathbb{E}_{\tilde \pi_t}g_t(x) \leq \frac{2\delta (x_t, \rho)} {\rho}.
		\end{align*}
	\end{lemma}
	
	\begin{proof}
		Consider the optimization problem  $\max_{\pi_t} \mathbb{E}_{\pi_t} f_t(x) $ $ \textit{\space s.t.\space} \mathbb{E}_{\pi_t} g_t(x) \leq 0 $.
		%\begin{align*}
		%    \max_{\pi_t} \mathbb{E}_{\pi_t} f_t(x) 
		%    \textit{\space s.t.\space} \mathbb{E}_{\pi_t} g_t(x) \leq 0.
		%\end{align*}
		Define the associated value function
		\begin{align*}
			v(\tau) = \max_{\pi_t \in \Delta } \{\mathbb{E}_{\pi_t} f_t(x) | \mathbb{E}_{\pi_t} g_t(x) \leq -\tau \}.
		\end{align*}
		
		By definition we have $v(0) = \mathbb{E}_{\pi_t^*} f_t(x)$. We break down the following analysis into four steps.
		
		Step 1): As $\phi^*_t$ is the optimal solution to the dual problem, following Theorem 3.59 in \cite{beck2017first}, we have $- \phi^*_t \in \partial v(0) $, which means $v(\tilde \tau ) - v(0) \leq  - \phi^*_t \tilde \tau$.
		
		%\begin{align*}
		%    v(\tilde \tau ) - v(0) \leq  - \phi^*_t \tilde \tau.
		%\end{align*}
		
		Step 2): We discuss the following  two cases.
		
		If $- \tilde \tau > 0$, we have the same result as Lemma 9 in  \cite{ding2021provably},
		\begin{align*}
			\mathbb{E}_{\tilde \pi_t}f_t(x) \leq \mathbb{E}_{\pi_t^*}f_t(x) = v(0) \leq v(\tilde \tau).
		\end{align*}
		
		If $- \tilde \tau < 0$, from the definition on  policy $\tilde \pi_t$, $\mathbb{E}_{\tilde \pi_t}f_t(x) = f_t(x_t)$ and $\mathbb{E}_{\tilde \pi_t}g_t(x) $$= g_t(x_t) $$= - \tilde \tau$, and the above definition of  $v(\tau)$, we have $\mathbb{E}_{\tilde \pi_t}f_t(x) \leq  v(\tilde \tau)$.
		%\begin{align*}
		%    \mathbb{E}_{\tilde \pi_t}f_t(x) \leq  v(\tilde \tau).
		%\end{align*}
		
		Step 3); Combining above two steps, we have 
		\begin{align*}
			\mathbb{E}_{\pi_t^*}f_t(x) - \phi^*_t \tilde \tau =  v(0) - \phi^*_t \tilde \tau \geq v(\tilde \tau) \geq  \mathbb{E}_{\tilde \pi_t}f_t(x). 
		\end{align*}
		
		Step 4): This yields 
		\begin{align*}
			(\rho - \phi^*_t ) (- \tilde \tau ) &=   \tilde \tau \phi^*_t  +  \rho (- \tilde \tau ) \\
			&\leq \mathbb{E}_{\pi_t^*}f_t(x) - \mathbb{E}_{\tilde \pi_t}f_t(x)   + \rho (- \tilde \tau ) \leq \delta(x_t, \rho).
		\end{align*}
		where the last line  uses the result in step 3) and Lemma \ref{lemma: no_xi_delta}.
		
		As $\rho \geq 2 \phi^*_t$ and $\mathbb{E}_{\tilde \pi_t}g_t(x) = g_t(x_t) = - \tilde \tau$, we have 
		\begin{align*}
			\mathbb{E}_{\tilde \pi_t}g_t(x) \leq \frac{\delta(x_t, \rho)}{\rho - \phi^*_t} \leq \frac{2 \delta(x_t, \rho)}{\rho}.
		\end{align*}
	\end{proof}

	As the optimal dual $\phi^*$ is within range $[0, \frac{2B}{\tau}]$, the above condition $\rho \geq 2 \phi^*_t$ is satisfied automatically when $\rho \geq \frac{4B}{\tau}$. %Finally, take a summation of $g_t(x_t)$ from 1 to $T$, we can bound $V(T)$ as follows. Two forms of bounds are established by setting restart period $W$ when $B_{\Delta}$ is known or not.

	\begin{lemma}
		%\yuntian{If we define $V(T)= \sum_{t=1}^T [g_t(x_t)]_+$, }
		The constraint violation is bounded as
		\begin{align*}
			V(T) \leq  O ((1+\frac{1}{\rho}) \hat \gamma_T^{\frac{7}{8}} B_{\Delta} T^{\frac{3}{4}} + \frac{G^2 T^{\frac{1}{2}}}{\rho} + \rho T^{\frac{1}{2}} ).
		\end{align*}
	\end{lemma}
	
	\begin{proof}
		We consider the constraint violation
		\begin{align*}
			V(T) &\leq \frac{2}{\rho} \sumt \delta(x_t, \rho)\\
			&\leq  \sumt  \Delta g_t  + \sumt  \Delta f_t / \rho + \sumt  \beta_t \sigma_{t-1}(x_t)/\rho \\
			&+ \sumt \tilde \beta_t \tilde \sigma_{t-1}(x_t) +  \eta G^2 T/\rho + \frac{1}{\eta} \rho.
		\end{align*}
		where $ \sumt \frac{1}{2\eta} [(\phi_t -\rho)^2 -(\phi_{t+1} - \rho)^2] \leq  \frac{1}{2\eta} (2\rho \phi_T ) \leq \frac{1}{\eta} \rho^2$.
		
		Let $\hat \gamma_T = \max\{ \gamma_T, \tilde \gamma_T \} $, $B_{\Delta}= \max \{B_g, B_f \}$, $\hat \beta_T = \max\{ \beta_T, \tilde \beta_T \} = O(\hat \gamma_T^{\frac{1}{2}})$, $W= \hat \gamma_T^{\frac{1}{4}} T^{\frac{1}{2}}$, $\eta=\frac{\rho}{G \sqrt{T}}$, we have,
		
		\begin{align*}
			V(T) \leq O \left(\left(1+\frac{1}{\rho}\right) \hat \gamma_T^{\frac{7}{8}} B_{\Delta} T^{\frac{3}{4}} + G\sqrt{T} \right).
		\end{align*}

		If $W= \hat \gamma_T^{\frac{1}{4}} T^{\frac{1}{2}} B_{\Delta}^{-\frac{1}{2}}$ \cite{zhou2021no}, we have a tighter bound $V(T) \leq O \left(\left(1+\frac{1}{\rho}\right) \hat \gamma_T^{\frac{7}{8}} B_{\Delta}^{\frac{1}{4}} T^{\frac{3}{4}} + G\sqrt{T} \right)$.
	\end{proof}

	%%%%%%%%%%%%%%%%%%%%%%%%%%%%%%%%
	\section{Conclusion}
	In this paper, we formulate the mmWave beam alignment problem in the time-varying multipath environment as a non-stationary kernelized bandit learning problem with constraints. The inherent correlation among the beams at successive time steps is captured by a kernel.  A primal-dual method is employed to tackle the constrained learning problem. Through periodic restarts, the proposed algorithm can adaptively adjust the beam to explore the environment and find the optimal beam with a high probability. Theoretical analysis demonstrates that both the received signal and the interference constraint converge to the optimal solution in the limit, and thus, this algorithm is asymptotically optimal. 
	For future work, it would be interesting to change the fixed restart period to an adaptive period through change detection \cite{liu2018change} -- thereby extending the algorithm in \cite{wei2021non} to our soft-constrained setting. It is unclear whether the Gaussian Process bandit can achieve the regret bound in \cite{wei2021non}  and how to incorporate constraints with two stationary tests in this algorithm. %We will also try to do real-time experiments to fit the parameters.
	
	\bibliographystyle{IEEEtran}
	\bibliography{sample}

\end{document}